\documentclass{article}
\usepackage[latin9]{inputenc}
\usepackage{color,soul,verbatim}
\usepackage{amsmath}
\usepackage{amssymb}
\usepackage{amsthm}
\usepackage{amscd}
\usepackage{amsfonts}

\makeatletter
\newtheorem{theorem}{Theorem}

\newtheorem{corollary}[theorem]{Corollary}

\newtheorem{definition}[theorem]{Definition}
\newtheorem{example}[theorem]{Example}

\newtheorem{lemma}[theorem]{Lemma}

\input{tcilatex}
\makeatother

\begin{document}

\title{Skew cyclic codes over $\mathbb{F}_{p}+u\mathbb{F}_{p}$}
\maketitle
\begin{center}
\author{{\small Reza Dastbasteh{*}, Seyyed Hamed Mousavi{*}{*},
Taher Abualrub{*}{*}{*}, Nuh Aydin{*}{*}{*}{*},} 
{\small and Javad Haghighat{*}{*}}} 
\end{center}
\maketitle
\begin{center}
{\small {*}Department of Mathematics, Sabanci university, Istanbul,
Turkey}\\
{\small  {*}{*}Department of Electrical Engineering, Shiraz
University of Technology, Shiraz, Iran}\\
{\small  {*}{*}{*} Department of Mathematics and Statistics,
American University of Sharjah, Sharjah, UAE}\\
{\small  {*}{*}{*}{*} Department of Mathematics and Statistics,
Kenyon College, Gambier, Ohio, USA}\\
{\small  E-mails:}\\
{\small  dastbasteh@sabanciuniv.edu, h.moosavi@sutech.ac.ir,
abualrub@aus.edu, aydinn@kenyon.edu, haghighat@sutech.ac.ir}
\end{center}

%
%

\begin{abstract}
In this paper, we study skew cyclic codes with arbitrary length over the
ring $R=\mathbb{F}_{p}+u\mathbb{F}_{p}$ where $p$ is an odd prime and $%
u^{2}=0$. We characterize all skew cyclic codes of length $n$ as left $%
R[x;\theta ]$-submodules of $R_{n}=R[x;\theta ]/\langle x^{n}-1\rangle $. We
find all generator polynomials for these codes and describe their minimal
spanning sets. Moreover, an encoding and decoding algorithm is presented for
skew cyclic codes over the ring $R$. Finally, based on the theory we
developed in this paper, we provide examples of codes with good parameters
over $F_{p}$ with different odd prime $p.$ In fact,
example 25 in our paper is a new ternary code in the class of quasi-twisted
codes. The other examples we provided are examples of optimal codes. 
\end{abstract}

\section{Introduction}

Cyclic codes are an important class of codes from both theoretical and
practical points of view. Traditionally, cyclic codes were studied over
finite fields. Recently, finite rings and their ideals are employed to
construct cyclic codes with good error detection and error correction
capabilities \cite{calderbank95,pless,lopez,bonnecaze,wolfman2}. These codes
have found applications in various areas including wireless sensor networks,
steganography and burst errors \cite{identification,burst}. 



Boucher, \textit{et. al}, generalized the notion of cyclic codes in \cite%
{Boucher2007}. They used generator polynomials in a non-commutative
polynomial ring called skew polynomial ring. They gave examples of skew
cyclic codes with Hamming distances larger than previously best known linear
codes of the same length and dimension. In \cite{taher2010}, Abualrub, 
\textit{et. al}, generalized the concept of skew cyclic codes to skew
quasi-cyclic codes. They constructed several new codes with Hamming
distances exceeding the Hamming distances of the previously best known
linear codes with comparable parameters. Other papers have appeared that
make use of various non-commutative rings to construct linear codes with
good parameters \cite{Boucher2008,boucher2011,gao2013,bhaintwal}.

Let $p$ be an odd prime number. In this paper, we are interested in studying
skew cyclic codes over the ring $R=\mathbb{F}_{p}\mathbb{+}u\mathbb{F}_{p}$
where $u^{2}=0$. Note that if we let $p=2,$ then the ring $\mathbb{F}_{2}%
\mathbb{+}u\mathbb{F}_{2}$ has only the trivial automorphism, and therefore
skew cyclic codes over this ring are exactly the classical cyclic codes
studied in \cite{siap2007}. The motivation behind studying skew cyclic codes
over this specific ring is that
 compared to the class of cyclic codes over $R$, the class of skew
cyclic codes is larger. This suggests that there is a better possibility of
finding codes with good parameters from skew cyclic codes over $R$.



The paper is organized as follows: In section \ref{Preliminaries}, we
discuss some properties of the skew polynomial ring $R[x;\theta].$ 
In section \ref{sec:Generator-polynomials-of}, we find the set of generator
polynomials for skew cyclic codes over the ring $R$. Section \ref%
{sec:Minimal-spanning-sets} studies minimal generating sets for these codes
and their cardinality. Section \ref{sub:The-Encoding-and} includes an
encoding and decoding algorithm for these codes. Section \ref{sec:Examples}
includes examples of linear codes over $\mathbb{F}_p$ obtained from skew
cyclic codes over $R$ by the help of a Gray map. Section \ref{sec:Conclusion}
includes the conclusion of our work and suggestions for future work.

\section{Preliminaries}

\label{Preliminaries}

Let $p$ be an odd prime number. Consider the Galois field $\mathbb{F}_{p}$
of order $p$ and the ring $R=\mathbb{F}_{p}\mathbb{+}u\mathbb{F}_{p}=\left\{
a+ub|~a,b\in\mathbb{F}_{p},\text{ with }u^{2}=0\right\} =\mathbb{F}_{p}\left[%
u\right]/\left\langle u^{2}\right\rangle .$ Denote the set of units of $%
\mathbb{F}_{p}$ by $\mathbb{F}_{p}^{\ast}=\mathbb{F}_{p}-\left\{ 0\right\} .$
Let $\theta$ be an automorphism of the ring $R$ with order $%
o(\theta)=|\langle\theta\rangle|=e>1.$ Then, every element in the finite
field $\mathbb{F}_{p}$ is fixed under $\theta.$ Hence, $\theta\left(a%
\right)=a$ for any $a\in\mathbb{F}_{p}.$ The next Lemma characterizes the
elements of the group $Aut(R).$

\begin{lemma}
Let $\theta\in Aut(R)$ and $a+ub\in R.$ Then $\theta\left(a+ub\right)$ $%
=a+usb$ for some $s\in\mathbb{F}_{p}^{\ast}$.
\end{lemma}


\begin{proof}
Let $\theta \in Aut(R)$ and suppose that $\theta \left( u\right) =r+us$ for
some $r,s\in \mathbb{F}_{p}.$ Then $u^{2}=0$ and 
\begin{eqnarray*}
0 &=&\theta \left( u^{2}\right) =\theta (u)\theta (u) \\
0 &=&\left( r+us\right) \left( r+us\right)=r^2+2urs \\
0 &=&r^{2}.
\end{eqnarray*}%
Hence $r=0$ and $\theta \left( u\right) =us$ for some $s\in \mathbb{F}%
_{p}^{\ast }$ $.$ Now, let $a+ub\in R.$ Then 
\begin{eqnarray*}
\theta \left( a+ub\right) &=&\theta (a)+\theta (u)\theta (b) \\
&=&a+usb.
\end{eqnarray*}
\end{proof}

\noindent One can show by induction that if $\theta\left(a+ub\right)$ $=a+usb
$, then $\theta^{i}\left(a+ub\right)$ $=a+us^{i}b$ for any positive integer $%
i.$


\begin{definition}
Let $\theta$ be an automorphism on $R.$ Define the skew polynomial set $R%
\left[x;\theta\right]$ to be 
\begin{equation*}
R[x;\theta]=\left\{ 
\begin{array}{c}
f(x)=a_{0}+a_{1}x+a_{2}x^{2}+\cdots+a_{n}x^{n}|~\text{ } \\ 
a_{i}\in\text{ }R\text{ for all }i=0,1,\ldots,n%
\end{array}%
\right\}
\end{equation*}
\end{definition}

\noindent where the addition of these polynomials is defined in the usual
way while multiplication $\ast$ is defined using the distributive law and the rule 
\begin{equation}
(ax^{i})\ast(bx^{j})=a\theta^{i}(b)x^{i+j}.  \tag{1}  \label{eq:1}
\end{equation}

The set $R[x;\theta ]$ is a non-commutative ring called the skew polynomial
ring with the usual addition of polynomials and multiplication defined as in
Equation (\ref{eq:1}) . Note that if $a,b\in \mathbb{F}_{p}[x],$ then $%
(ax^{i})\ast(bx^{j})=a\theta ^{i}(b)x^{i+j}=abx^{i+j}$ because $\theta
\left( b\right) =b$ for all $b\in \mathbb{F}_{p}.$ Hence, the ring $\mathbb{F%
}_{p}[x]$ is a subring of $R[x;\theta ].$ 

\begin{theorem}
\cite{Mcdonald}\label{Division}(The Right Division Algorithm) Let $f$ and $g$
be two polynomials in $R[x;\theta ]$ 
\textcolor{black}{$\ $with
the leading coefficient of $f$ ~being unit}. Then there exist unique
polynomials $q$ and $r$ such that 
\begin{equation*}
g=q\ast f+r\text{ where }r=0\text{ or }\deg (r)<\deg (f).
\end{equation*}
\end{theorem}

The above result is called division on the right by $f.$ A similar result
can be proved regarding division on left by $f.$

\begin{theorem}
The center of $R[x;\theta]$ is the set $Z\left(R[x;\theta]\right)=\mathbb{F}%
_{p}[x^{e}]$ for any $\theta\in Aut(R)$ of order $e$.
\end{theorem}

\begin{proof}
The proof is similar to Lemma 1.1 in \cite{gao2013}. 
\end{proof}

As a result of this Theorem, the following corollary is clear.

\begin{corollary}
$x^{n}-1\in Z\left(R[x;\theta]\right)$ if and only if $e|n$.
\end{corollary}


The above corollary shows that the polynomial $\left(x^{n}-1\right)$ is in
the center $Z(R[x;\theta])$ of the ring $R[x;\theta]$, hence generates a
two-sided ideal if and only if the $e|n$. Consequently, the quotient space $%
R_{n}=R[x;\theta]/\langle x^{n}-1\rangle$ is a ring if and only if $e|n$. In
this case skew cyclic codes can be regarded as (left) ideals in $R_{n}$. In
this paper we are interested in skew cyclic codes for any length $n$
regardless of whether $e|n$ or not. We show that regarding them as modules
rather than ideals gives us the flexibility to handle skew cyclic codes of
all lengths in the same way.

Let $r(x)\in R[x;\theta \dot{]}$ and $\left( f(x)+\langle x^{n}-1\rangle
\right) \in R_{n}.$ Define%
\begin{equation*}
r(x)\left( f(x)+\langle x^{n}-1\rangle \right) =r(x)f(x)+\langle
x^{n}-1\rangle
\end{equation*}%
This multiplication with the usual addition leads to the following Lemma

\begin{lemma}
The quotient space $R_{n}$ is a left $R[x;\theta ]$ module.
\end{lemma}


\begin{definition}
Let $R$ be the ring $\mathbb{F}_{p}\mathbb{+}u\mathbb{F}_{p}$ and $\theta$
be an automorphism of \ $R$ with $\left\vert \left\langle
\theta\right\rangle \right\vert =e.$ A subset $C$ of $R^{n}$ is called a
skew cyclic code of length $n$ if $C$ satisfies the following conditions:
\end{definition}

\begin{enumerate}
\item $C$ is a submodule of $R^{n}.$

\item If $c=\left(c_{0},c_{1},\ldots,c_{n-1}\right)\in C,$ then so is its
the skew cyclic shift, i.e., $\left(\theta\left(c_{n-1}\right),\theta%
\left(c_{0}\right),\ldots,\theta\left(c_{n-2}\right)\right)\in C.$
\end{enumerate}

We have the usual representation of vectors $(c_{0},c_{1},\ldots,c_{n-1})\in
R^{n}$ by polynomials $c(x)=c_{0}+c_{1}x+\cdots+c_{n-1}x^{n-1}$. With this
identification, the skew cyclic shift of a codeword $%
c(x)=c_{0}+c_{1}x+c_{2}x^{2}+\cdots+c_{n-1}x^{n-1}\in C$ corresponds to $%
x\ast c(x)\func{mod}\left(x^{n}-1\right)$ which is equal to $%
\theta(c_{n-1})+\theta(c_{0})x+\cdots+\theta\left(c_{n-2}\right)x^{n-1}$.

As is common in the discussion of cyclic codes, we can regard codewords of a
skew cyclic code $C$ as vectors or as polynomials interchangeably. In either
case, we use the same notation $C$ ~to denote the set of all codewords. We
follow this convention in the definition below and in the rest of the paper.

\begin{definition}
\label{Polynomial}(Polynomial definition of skew cyclic codes)

A subset $C\subseteq R_{n}$ is called a skew cyclic code if $C$ satisfies
the following conditions:

\begin{enumerate}
\item $C$ is an $R$-submodule of $R_{n}$

\item If $c(x)=\left(a_{0}+a_{1}x+\ldots+a_{n-1}x^{n-1}\right)\in C$, then

$x\ast c(x)=$ $\left(\theta(a_{n-1})+\theta(a_{0})x+\cdots+%
\theta(a_{n-2})x^{r-1}\right)\in C.$
\end{enumerate}
\end{definition}


As a result of this definition, we get the following Lemma.

\begin{lemma}
\label{submodule}$C$ is a skew cyclic code of length $n$ over $R$ if and
only if $C$ is a left $R[x;\theta ]-$submodule of $R_{n}=R[x;\theta
]/\langle x^{n}-1\rangle $.
\end{lemma}

\section{\label{sec:Generator-polynomials-of}Generator polynomials of Skew
Cyclic Codes over $R$}

In this section we are interested in studying algebraic structures of skew
cyclic codes over $R.$ Using Lemma \ref{submodule}, our goal is to find the
generator polynomials of these codes as left $R[x;\theta]$-submodules of $%
R_{n}=R[x;\theta]/\langle x^{n}-1\rangle$.


\begin{lemma}
\label{partaker} For any $g(x)\in R[x;\theta ]$, there exists a unique $%
g^{\prime}(x)\in\mathbb{F}_{p}[x]$ of the same degree as $g(x)$ such that $%
g(x)\ast u=ug(x)^{\prime}$.
\end{lemma}


\begin{proof}
Let $g(x)=\sum_{i=0}^{n}g_{i}x^{i}\in R[x,\theta]$. Since $g_{i}\in R$ for
each $i$, there exist $g^{\prime }_{i},g^{\prime \prime }_{i}\in\,\mathbb{F}%
_{p}$ such that $g_{i}=g^{\prime }_{i}+ug^{\prime \prime }_{i}$. So $%
g(x)=\sum(g^{\prime }_{i}+ug^{\prime \prime }_{i})x^{i}$.

Now since $\alpha^{i}u=\theta^{i}(u)x=x^{i}u$, we can see

\begin{align*}
g(x)\ast u= & \left(\sum(g^{\prime }_{i}+ug^{\prime \prime }_{i})x^{i}\right)\ast u=\left(\sum
g^{\prime }_{i}x^{i}\right)\ast u+\left(\sum ug^{\prime \prime }_{i}x^{i}\right)\ast u \\
= & \sum g^{\prime }_{i}\theta^{i}(u)x^{i}+\sum ug^{\prime \prime
}_{i}\theta^{i}(u)x^{i}=\sum g^{\prime }_{i}\alpha^{i}ux^{i}+\sum ug^{\prime
\prime }_{i}\alpha^{i}ux^{i}.
\end{align*}

Obviously, the second sum is 0 (it contains a factor of $u^{2}$). Therefore, we have 
\begin{align}
g(x)\ast u=\sum g^{\prime }_{i}\alpha^{i}ux^{i}=u\sum g^{\prime
}_{i}\alpha^{i}x^{i}=ug^{\prime }(x).
\end{align}
The uniqueness of $g^{\prime }$ in $\mathbb{F}_{p}[x]$ is clear by this
proof.
\end{proof}

\textbf{Notation}. For a fixed element $g\in\mathbb{F}_{p}[x]$, the element $%
g^{\prime}\in\mathbb{F}_{p}[x]$ in Lemma \ref{partaker} is unique and hence
we call it the \textit{partaker} of $g$. Also note that if $%
g=g_{1}+ug_{2}\in R[x;\theta],$ then $ug=u\left(g_{1}+ug_{2}%
\right)=ug_{1}=g^{\prime}u.$

\begin{example}
Suppose $g(x)=1+x+x^{2}\in (\mathbb{F}_p+u\mathbb{F}_p)[x;\theta]$ where $\theta(a+ub)=a+\alpha ub$ for $\alpha\in \mathbb{F}_p$. Then $g(x)\ast
u=(1+x+x^{2})\ast u=u+x\ast u+x^{2}\ast
u=u+\theta(u)x+\theta^{2}(u)x^{2}=u+\alpha ux+\alpha^{2}ux^{2}$.

Therefore, $g(x)u=u(1+\alpha x+\alpha^{2}x^{2})$. So $g^{\prime}=%
\alpha^{2}x^{2}+\alpha x+1$.

Note that there are infinitely many elements $h\in R$ such that $gu=uh$. It is sufficient to define $%
h=g^{\prime}+ul$ ~for every $l\in\mathbb{F}_{p}[x]$.
\end{example}

\begin{lemma}
\label{factorization}The polynomial $x^{n}-1$ factors in the ring $\mathbb{F}%
_{p}[x]$ as $x^{n}-1$ $=f_{1}(x)g_{1}(x)$ if and only if $x^{n}-1=f(x)\ast
g(x)$ in the ring $R\left[x;\theta\right]$ where $f(x)=f_{1}(x)+uf_{2}(x)$
and $g(x)=g_{1}(x)+ug_{2}(x)$ for some polynomials $f_{2}(x),~g_{2}(x)$ in $%
\mathbb{F}_{p}[x]$.
\end{lemma}

\begin{proof}
$\Rightarrow)$ Suppose $x^{n}-1$ $=f_{1}(x)g_{1}(x)$ in the ring $\mathbb{F}%
_{p}[x]$. Since $\mathbb{F}_{p}\left[x\right]$ is a subring of $R[x;\theta]$, $x^{n}-1$ $=f_{1}(x)g_{1}(x)$ in $R[x;\theta].$ Hence, we let $f_{2}(x)=g_{2}(x)=0.$

$\Leftarrow)$ Suppose $x^{n}-1=f(x)\ast g(x)$ in the ring $R\left[x;\theta\right]$
where $f(x)=f_{1}(x)+uf_{2}(x)$ and $g(x)=g_{1}(x)+ug_{2}(x)$ for some
polynomials $f_{2}(x),~g_{2}(x)$ in $\mathbb{F}_{p}[x].$ Then 
\begin{eqnarray*}
x^{n}-1 & = & f(x)\ast g(x) \\
& = & \left(f_{1}(x)+uf_{2}(x)\right)\ast \left(g_{1}(x)+ug_{2}(x)\right) \\
& = & f_{1}(x)g_{1}(x)+f_{1}(x)\ast ug_{2}(x)+uf_{2}(x)g_{1}(x).
\end{eqnarray*}
Using Lemma \ref{partaker}, we know that $f_{1}(x)u=uk(x)$ for some $k(x)\in%
\mathbb{F}_{p}[x].$ Hence, 
\begin{eqnarray*}
x^{n}-1 & = & f(x)\ast g(x) \\
& = & f_{1}(x)g_{1}(x)+f_{1}(x)\ast ug_{2}(x)+uf_{2}(x)g_{1}(x) \\
& = & f_{1}(x)g_{1}(x)+uk(x)g_{2}(x)+uf_{2}(x)g_{1}(x) \\
& = & f_{1}(x)g_{1}(x)+u\left(k(x)g_{2}(x)+f_{2}(x)g_{1}(x)\right).
\end{eqnarray*}
Suppose $x^{n}-1$ $=f_{1}(x)g_{1}(x)+r_{1}(x)$ in the ring $\mathbb{F}%
_{p}[x].$ Since $\mathbb{F}_{p}[x]$ is a subring of $R\left[x;\theta\right]$,
$x^{n}-1$ $=f_{1}(x)g_{1}(x)+r_{1}(x)$ in the ring $R\left[x;\theta%
\right]$ as well. Hence, 
\begin{eqnarray*}
x^{n}-1 & = & f_{1}(x)g_{1}(x)+u\left(k(x)g_{2}(x)+f_{2}(x)g_{1}(x)\right) \\
& = & x^{n}-1-r_{1}(x)+u\left(k(x)g_{2}(x)+f_{2}(x)g_{1}(x)\right) \\
r_{1}(x) & = & u\left(k(x)g_{2}(x)+f_{2}(x)g_{1}(x)\right).
\end{eqnarray*}
But $r_{1}(x)\in\mathbb{F}_{p}[x].$ This is a contradiction unless $%
r_{1}(x)=0$ and then $x^{n}-1$ $=f_{1}(x)g_{1}(x)$ in the ring $\mathbb{F}%
_{p}[x].$
\end{proof}

Note that in the above Lemma $f_{1}(x)$ and $g_{1}(x)$ are unique
polynomials because the ring $\mathbb{F}_{p}[x]$ $\ $is a unique
factorization ring; however, $f_{2}(x)$ and $g_{2}(x)$ are not unique. This
is justified by noting that the ring $R\left[ x;\theta \right] $ is not a
unique factorization ring. We know that \label{units1}$U(R)=\mathbb{F}%
_{p}^{\ast }+u\mathbb{F}_{p}.$ Based on this fact, we try to find $%
U(R[x;\theta ]).$


\begin{lemma}
\label{units2}$U\left(R\left[x;\theta\right]\right)=\left\{ a+uh(x)|a\in%
\mathbb{F}_{p}^{\ast},\text{ and }h(x)\in\mathbb{F}_{p}\left[x\right]%
\right\} $
\end{lemma}

\begin{proof}
Let $a+uh(x)\in R\left[ x;\theta \right] $ where $a\in \mathbb{F}_{p}^{\ast }
$ and $h(x)\in \mathbb{F}_{p}\left[ x\right] .$ Then%
\begin{eqnarray*}
\left( a+uh\right) \left( a^{-1}-a^{-1}uha^{-1}\right) 
&=&1-uha^{-1}+uha^{-1}-uha^{-1}uha^{-1} \\
&=&1-uha^{-1}uha^{-1}=1-u^{2}h_{1}a^{-1}ha^{-1}\text{ (using Lemma \ref%
{partaker})} \\
&=&1.
\end{eqnarray*}%
Hence, $a+uh(x)\in U\left( R\left[ x;\theta \right] \right) $. Conversely,
let $f\in U\left( R\left[ x;\theta \right] \right) $. Then, there exists $%
g\in R\left[ x;\theta \right] $ such that $fg=gf=1$. Let $f=f_{1}+uf_{2}$
and $g=g_{1}+ug_{2}$ for $f_{i},g_{i}\in \mathbb{F}_{p}[x]$. Then, $%
fg=(f_{1}+uf_{2})(g_{1}+ug_{2})=1$ implies that $f_{1}g_{1}=1$ and $%
uf_{2}g_{1}+f_{1}ug_{2}=0$. Hence, $f_{1}$ is a non-zero constant
polynomial. That is, $f_{1}\in \mathbb{F}_{p}^{\ast }$. Thus, $f=f_{1}+uf_{2}
$, where $f_{1}\in \mathbb{F}_{p}^{\ast }$ and $f_{2}\in \mathbb{F}_{p}[x]$.
\end{proof}

Let $C$ be be a nonzero skew cyclic code over $R$ and let $%
c(x)=\left(a_{0}+a_{1}x+\ldots+a_{n-1}x^{n-1}\right)\in C.$ If $a_{n-1}$ is
a unit in $R$ with inverse $w$ then $wc(x)$ is a monic polynomial in $C.$
Hence, for any nonzero skew cyclic code we have the following cases to
consider:

\textbf{Case1:} $C$ has no monic polynomials

\textbf{Case2:} $C$ has at least one monic polynomial.

The next lemma classifies all skew cyclic codes that satisfy Case 1.

\begin{lemma}
\label{nonmonic}Let $C$ be a nonzero skew cyclic code that has no monic
polynomials. Then $C=\left\langle u\overline{a(x)}\right\rangle $ where $%
\overline{~a(x)}$ is a polynomial of minimal degree in $C$ and $x^{n}-1=%
\overline{b(x)}~\overline{a(x)}$ in $\ \mathbb{F}_{p}[x].$
\end{lemma}

\begin{proof}
Suppose $C$ is a nonzero skew cyclic code that has no monic polynomials and
suppose that 
\begin{equation*}
\theta\left(a+ub\right)=a+usb\text{ where }s\in\mathbb{F}_{p}^{\ast}. 
\end{equation*}
Let 
\begin{equation*}
a(x)=a_{0}+a_{1}x+\ldots+u\overline{a_{r}}x^{r} 
\end{equation*}
be a polynomial of minimal degree in $C$ where $\overline{a_{r}}\in\mathbb{F}%
_{p}^{\ast}$ and $a_{i}\in R$ for all $i=0,1,\ldots,r-1.$ Note that 
\begin{equation*}
ua(x)=ua_{0}+ua_{1}x+\ldots+ua_{r-1}x^{r-1}\in C. 
\end{equation*}
Since $a(x)$ is of minimal degree in $C,$  $ua(x)=0$ and 
\begin{equation*}
a(x)=u\overline{a(x)}, 
\end{equation*}
where $\overline{a(x)}\in\mathbb{F}_{p}\left[x\right]$ and $a_{i}=u\overline{%
a_{i}}$ for all $i=0,1,\ldots,r.$ Let $c(x)$ be any codeword in $C.$ Then, $%
c(x)$ is not monic. Hence, $c(x)=c_{0}+c_{1}x+\ldots+u\overline{c_{t}}x^{t}$
where $t\leq n-1,~\overline{c_{t}}\in\mathbb{F}_{p}^{\ast}$ and $c_{i}\in R.$
for all $i=0,1,\ldots,t-1.$ We want to prove that $c(x)=u\overline{c(x)}.$
Write $c(x)=c_{1}(x)+c_{2}(x)$ where all terms in $c_{1}(x)$ have powers
less than $r$ while all terms in $c_{2}(x)$ have powers larger or equal $r.$
Suppose $c_{t-1}$ is a unit. Note that $\theta^{i}\left(u\right)=us^{i}.$
Consider the polynomial $Z(x)=z_{1}(x)-z_{2}(x)\in C,$ where 
\begin{eqnarray*}
z_{1}(x) & = & \left(s^{t-r}\right)^{-1}\left(a_{r}\right)^{-1}x^{t-r}a(x) \\
& = &
\left(s^{t-r}\right)^{-1}\left(a_{r}\right)^{-1}a_{0}\theta^{^{t-r}}\left(u%
\right)x^{^{t-r}}+\left(s^{^{t-r}}\right)^{-1}\left(a_{r}\right)^{-1}a_{1}%
\theta^{t-r}\left(u\right)x^{t-r+1}+ \\
& &
\ldots+\left(s^{^{t-r}}\right)^{-1}\left(a_{r}\right)^{-1}a_{r-1}%
\theta^{^{t-r}}\left(u\right)x^{t-1}+\left(s^{^{t-r}}\right)^{-1}\left(a_{r}%
\right)^{-1}a_{r}\theta^{^{t-r}}\left(u\right)x^{t} \\
& = &
\left(s^{^{t-r}}\right)^{-1}\left(a_{r}\right)^{-1}a_{0}\theta^{^{t-r}}%
\left(u\right)x^{^{t-r}}+\left(s^{^{t-r}}\right)^{-1}\left(a_{r}%
\right)^{-1}a_{1}\theta^{^{t-r}}\left(u\right)x^{t-r+1}+ \\
& &
\ldots+\left(s^{^{t-r}}\right)^{-1}\left(a_{r}\right)^{-1}a_{r-1}%
\theta^{^{t-r}}\left(u\right)x^{t-1}+ux^{t},
\end{eqnarray*}
and 
\begin{equation*}
z_{2}(x)=\left(c_{t}\right)^{-1}c(x)=\left(c_{t}\right)^{-1}c_{0}+%
\left(c_{t}\right)^{-1}c_{1}x+\ldots+\left(c_{t}%
\right)^{-1}c_{t-1}x^{t-1}+ux^{t}. 
\end{equation*}
Hence, $Z(x)=z_{1}(z)-z_{2}(z)$ is a polynomial of $\deg(t-1)$ in $C$ where
the coefficient of $x^{t-1}$ is $z_{t-1}=\left(s^{^{t-r}}\right)^{-1}%
\left(a_{r}\right)^{-1}a_{r-1}\theta^{^{t-r}}\left(u\right)-\left(c_{t}%
\right)^{-1}c_{t-1}=\eta u-\left(c_{t}\right)^{-1}c_{t-1}$ where $%
\left(c_{t}\right)^{-1}c_{t-1}$ is a unit. By Lemma \ref{units1}, $z_{t-1}$
is a unit and hence $C$ has a monic polynomial. This is a contradiction
since $C$ has no monic polynomials. Using the same procedure we can show
that $c_{i}$ is not a unit for all $c_{i}\in c_{2}(x).$ Suppose that $c_{i}$
is a unit for some $i$ in $c_{1}(x).$ Then $uc(x)=uc_{1}(x)\in C$ and $%
uc_{1}(x)$ is a a nonzero polynomial with $\deg uc_{1}(x)<\deg a(x).$ Again,
this is a contradiction. Hence, 
\begin{equation*}
c(x)=u\overline{c(x)}, 
\end{equation*}
where $\overline{c(x)}\in\mathbb{F}_{p}[x]$ and $c_{i}=u\overline{c_{i}}$
for all $i=0,1,\ldots,t.$ Since $\overline{a(x)}$ and $\overline{c(x)}$ are
two polynomials in $\mathbb{F}_{p}[x],$  by the division algorithm,
there exist polynomials $q(x),r(x)$ in $\mathbb{F}_{p}[x]$ such that 
\begin{equation*}
\overline{c(x)}=q(x)\overline{a(x)}+r(x) 
\end{equation*}
where $r(x)=0$ or $\deg r(x)<\deg\overline{a(x)}=\deg a(x).$ Hence, using
Lemma \ref{partaker}, we get 
\begin{eqnarray*}
u\overline{c(x)} & = & uq(x)\overline{a(x)}+ur(x) \\
& = & q^{\prime}(x)u\overline{a(x)}+ur(x).
\end{eqnarray*}
This implies that 
\begin{equation*}
ur(x)=u\overline{c(x)}-q^{\prime}(x)u\overline{a(x)}\in C. 
\end{equation*}
This is a contradiction because $\deg ur(x)<\deg\overline{a(x)}=\deg a(x).$
Therefore, $ur(x)=0.$ Since $r(x)\in\mathbb{F}_{p}[x],$ then $r(x)=0$ and 
\begin{equation*}
u\overline{c(x)}=q^{\prime}(x)u\overline{a(x)}. 
\end{equation*}
Hence, $C=\left\langle u\overline{a(x)}\right\rangle .$ Again, since $%
\overline{a(x)}$ and $x^{n}-1$ are polynomials in $\mathbb{F}_{p}[x],$ then
by the division algorithm, there exist polynomials$\overline{\text{ }b(x)}%
,r_{1}(x)$ in $\mathbb{F}_{p}[x]$ such that 
\begin{equation*}
x^{n}-1=\overline{b(x)}~\overline{a(x)}+r_{1}(x), 
\end{equation*}
where $r_{1}(x)=0$ or $\deg r_{1}(x)<\deg\overline{a(x)}=\deg a(x).$ Hence, 
\begin{eqnarray*}
u\left(x^{n}-1\right) & = & u\text{ }\overline{b(x)}~\overline{a(x)}%
+ur_{1}(x) \\
& = & q_{1}^{\prime}(x)u\overline{a(x)}+ur_{1}(x).
\end{eqnarray*}
In the ring $R_{n}=R[x;\theta]/\langle x^{n}-1\rangle,$ we get that 
\begin{equation*}
0=q_{1}^{\prime}(x)u\overline{a(x)}+ur_{1}(x) 
\end{equation*}
or, 
\begin{equation*}
ur_{1}(x)=-q_{1}^{\prime}(x)u\overline{a(x)}\in C. 
\end{equation*}
A contradiction. Hence, $ur_{1}(x)=0.$ Since $r_{1}(x)\in\mathbb{F}_{p}[x],$
 $r_{1}(x)=0$ and 
\begin{equation*}
x^{n}-1=\overline{b(x)}~\overline{a(x)}. 
\end{equation*}
\end{proof}

\begin{lemma}
\label{monic1}Let $C$ be a nonzero skew cyclic code that has at least one
monic polynomial and let $g(x)$ be a polynomial of minimal degree in $C.$
Suppose that $g(x)$ is monic. Then $C=\left\langle g(x)\right\rangle $ where 
$x^{n}-1=k(x)g(x)$ in $R_{n}.$
\end{lemma}

\begin{proof}
The proof is a straightforward application of Theorem \ref{Division}.
\end{proof}

\begin{lemma}
\label{monic2}Let $C$ be a nonzero skew cyclic code that has at least one
monic polynomial. Moreover, suppose that all polynomials of minimal degree
are not monic. Let $a(x)$ be a polynomial of minimal degree in $C.$ Let $%
g(x) $ be a monic polynomial in $C$ of minimal degree among all monic
polynomials in $C.$ Then, $C=\left\langle g(x)+up(x),a(x)\right\rangle
=\left\langle g(x)+up(x),u\overline{a(x)}\right\rangle ,$ where $%
x^{n}-1=k(x)\ast g(x)$ in $R[x;\theta],$ $x^{n}-1=\overline{b(x)}~\overline{%
a(x)}$ in $\mathbb{F}_{p}[x] $ and $\deg a(x)<\deg g(x).$
\end{lemma}

\begin{proof}
Let $C$ ~be a skew cyclic code that has at least one monic polynomial and
let $a(x)$ be a polynomial of minimal degree in $C.$ Since $a(x)$ is not
monic and of minimal degree in $C,$ as in the proof of Lemma \ref%
{nonmonic}, we can show that $a(x)=$ $u\overline{a(x)}$ where $\overline{a(x)%
}\in\mathbb{F}_{p}\left[x\right].$ Suppose that $c(x)$ is a codeword in $C.$
Let $f(x)$ be a monic polynomial of minimal degree in $C.$ Then using
Theorem \ref{Division}, there exist two polynomials $q_{2}(x)$ and $r_{2}(x)$
in $R_{n}$ such that 
\begin{equation*}
c(x)=q_{2}(x)\ast f(x)+r_{2}(x) 
\end{equation*}
where $r_{2}(x)=0$ or $\deg r_{2}(x)<\deg f(x).$ Then 
\begin{equation*}
r_{2}(x)=c(x)-q_{2}(x)\ast f(x)\in C. 
\end{equation*}
Since $f(x)$ is a monic polynomial of minimal degree in $C,$  $r_{2}(x)$
is not monic. In fact, as in the proof of Lemma \ref{nonmonic}, one can
easily show that $r_{2}(x)=ur_{3}(x)$ for some polynomial $r_{3}(x)\in%
\mathbb{F}_{p}[x].$ Now, apply the division algorithm on $\overline{a(x)}$
and $r_{3}(x)$ to obtain 
\begin{equation*}
r_{3}(x)=q_{3}(x)\overline{a(x)}+r_{4}(x), 
\end{equation*}
where $r_{4}(x)=0$ or $\deg r_{4}(x)<\deg\overline{a(x)}.$ Hence, 
\begin{eqnarray*}
r_{2}(x) & = & ur_{3}(x)=uq_{3}(x)\overline{a(x)}+ur_{4}(x) \\
& = & q_{4}u\overline{a(x)}+ur_{4}(x).
\end{eqnarray*}
Thus, $ur_{4}(x)\in C.$ Since $\deg r_{4}(x)<\deg\overline{a(x)}=\deg a(x)$, $ur_{4}(x)=0$ and 
\begin{equation*}
r_{2}(x)=ur_{3}(x)=uq_{3}(x)\overline{a(x)}=q_{4}u\overline{a(x)}. 
\end{equation*}
Therefore, 
\begin{eqnarray*}
c(x) & = & q_{2}(x)f(x)+r_{2}(x) \\
& = & q_{2}(x)f(x)+q_{4}u\overline{a(x)}.
\end{eqnarray*}
Thus $C=\left\langle f,a(x)\right\rangle =\left\langle f,u\overline{a(x)}%
\right\rangle .$ Note that $x^{n}-1$ and $\overline{a(x)}$ are polynomials
in $\mathbb{F}_{p}[x].$ Hence, by the division algorithm we can write 
\begin{equation*}
x^{n}-1=q_{5}\overline{a(x)}+r_{5} 
\end{equation*}
where $r_{5}(x)=0$ or $r_{5}(x)$ is a polynomial in $\mathbb{F}_{p}[x]$ with 
$\deg r_{5}(x)<\deg\overline{a(x)}.$ Then we have 
\begin{eqnarray*}
u\left(x^{n}-1\right) & = & uq_{5}\overline{a(x)}+ur_{5} \\
& = & q_{5}^{\prime}u\overline{a(x)}+ur_{5.}
\end{eqnarray*}
In $R_{n},$ we obtain 
\begin{equation*}
ur_{5}(x)=-q_{5}^{\prime}u\overline{a(x)}\in C 
\end{equation*}
with $\deg ur_{5}(x)=\deg r_{5}(x)<\deg\overline{a(x)}.$ This is a
contradiction unless $ur_{5}(x)=0.$ Since $r_{5}(x)\in\mathbb{F}_{p}[x]$,
we have $r_{5}(x)=0$ and 
\begin{equation*}
x^{n}-1=q_{5}\overline{a(x)} 
\end{equation*}
Moreover, let $f(x)=f_{1}(x)+uf_{2}(x)$ where $f_{1}(x),~f_{2}(x)\in\mathbb{F%
}_{p}[x].$ Then using Theorem \ref{Division}, we have 
\begin{equation*}
x^{n}-1=q_{3}(x)\left(f_{1}(x)+uf_{2}(x)\right)+R(x), 
\end{equation*}
where $R(x)=0$ or $\deg R(x)<\deg f(x)=\deg f_{1}(x).$ This implies that $%
R(x)\in C.$ Since $f(x)$ is a monic polynomial of minimal degree in $C$, we conclude that $R(x)$ is not monic. Hence, $R(x)=w(x)u\overline{a(x)}%
=uw_{1}(x)\overline{a(x)}$ (by Lemma \ref{partaker}). Thus 
\begin{eqnarray*}
x^{n}-1 & = & q_{3}(x)\left(f_{1}(x)+uf_{2}(x)\right)+uw_{1}(x)\overline{a(x)%
} \\
& = & q_{3}(x)f_{1}(x)+q_{3}(x)uf_{2}(x)+uw_{1}(x)\overline{a(x)} \\
& = & q_{3}(x)f_{1}(x)+uq_{4}(x)f_{2}(x)+uw_{1}(x)\overline{a(x)}.
\end{eqnarray*}
Hence, in the ring $\mathbb{F}_{p}[x],$ we have 
\begin{equation*}
x^{n}-1=q_{3}(x)f_{1}(x). 
\end{equation*}
By Lemma \ref{factorization}, there must be a polynomial $g(x)$ in $%
R[x;\theta]$ of degree $f_{1}(x)$ such that $g(x)$ is a right divisor of $%
x^{n}-1$ in $R[x;\theta].$ Hence, $g(x)=f_{1}(x)+ul_{1}(x)$ where $l_{1}(x)$
is a polynomial in $\mathbb{F}_{p}[x]$ of degree less than the degree of $%
f_{1}(x).$ Thus, 
\begin{eqnarray*}
f(x) & = & f_{1}(x)+uf_{2}(x) \\
& = & g(x)-ul_{1}(x)+uf_{2}(x) \\
& = & g(x)+u\left(f_{2}(x)-l_{1}(x)\right).
\end{eqnarray*}
Therefore, $C=\left\langle f,a(x)\right\rangle =\left\langle f,u\overline{%
a(x)}\right\rangle =\left\langle g(x)+u\left(f_{2}(x)-l_{1}(x)\right),u%
\overline{a(x)}\right\rangle $ with $x^{n}-1=k(x)g(x)$ in $R[x;\theta]$ and $%
x^{n}-1=\overline{b(x)}~\overline{a(x)}$ in $\mathbb{F}_{p}[x]$ and $\deg%
\overline{a(x)}<\deg g(x).$
\end{proof}

\begin{lemma}
\label{monic3}Let $C=\left\langle g(x)+up(x),a(x)\right\rangle =\left\langle
g(x)+up(x),u\overline{a(x)}\right\rangle $ as in Lemma \ref{monic2}. Then $%
\overline{a(x)}|g(x)\func{mod}u$ and $\dfrac{x^{n}-1}{g}up\in\left\langle
ua\right\rangle .$
\end{lemma}

\begin{proof}
Suppose $C=\left\langle g(x)+up(x),a(x)\right\rangle =\left\langle
g(x)+up(x),u\overline{a(x)}\right\rangle $ as in Lemma \ref{monic2}. Let $%
uc(x)\in C$ where $c(x)\in\mathbb{F}_{p}[x].$ Using the division algorithm
one can write 
\begin{equation*}
c(x)=q(x)\overline{a(x)}+r(x), 
\end{equation*}
where $r(x)=0$ or $\deg r(x)<\deg\overline{a(x)}.$ Hence, 
\begin{eqnarray*}
uc(x) & = & uq(x)\overline{a(x)}+ur(x) \\
& = & q^{\prime}(x)u\overline{a(x)}+ur(x).
\end{eqnarray*}
This implies that $ur(x)\in C.$ This is a contradiction since $\deg
ur(x)=\deg r(x)<\deg a(x)=\deg\overline{a(x)}.$ Hence, $ur(x)=0$ and $r(x)=0$
because $r(x)\in\mathbb{F}_{p}[x].$ This implies $uc(x)=uq(x)\overline{%
a(x)}=q^{\prime}(x)u\overline{a(x)}\in\left\langle ua\right\rangle .$
Therefore, for any polynomial of the form $uc(x)\in C,$ we have $c(x)=q(x)%
\overline{a(x)}$ and $uc(x)\in\left\langle ua\right\rangle$. Since 
$u\left(g(x)+up(x)\right)=ug(x)\in C$,  we get that $\overline{a(x)}|g(x)%
\func{mod}u.$ Also $\dfrac{x^{n}-1}{g}\left(g(x)+up(x)\right)=\dfrac{x^{n}-1%
}{g}up(x)=u\left(\dfrac{x^{n}-1}{g}\right)^{\prime}p(x)\in C.$ Hence $\dfrac{%
x^{n}-1}{g}up\in\left\langle ua\right\rangle $ $.$
\end{proof}

We summarize the results of Lemmas \ref{nonmonic}, \ref{monic1}, \ref{monic2}
and \ref{monic3} in the following theorem that classifies all skew cyclic
codes over the ring $R.$

\begin{theorem}
\label{classifications}Let $C$ be a nonzero skew cyclic code over the ring $%
R.$ Then $C$ satisfies one of the following cases:

\begin{enumerate}
\item $C$ has no monic polynomials. Then $C=\left\langle u\overline{a(x)}%
\right\rangle $ where $\overline{~a(x)}$ is a polynomial of minimal degree
in $C$ and $x^{n}-1=\overline{b(x)}~\overline{a(x)}$ in $\ \mathbb{F}%
_{p}[x]. $

\item $C$ has a monic polynomial $g(x)$ of minimal degree in $C.$ Then $%
C=\left\langle g(x)\right\rangle $ where $g(x)$ is a polynomial of minimal
degree in $C$ and $x^{n}-1=k(x)\ast g(x)$ in $R_{n}.$

\item All monic polynomials in $C$ are not of minimal degree. Then we have $%
C=\left\langle g(x)+up(x),a(x)\right\rangle =\left\langle g(x)+up(x),u%
\overline{a(x)}\right\rangle ,$ where $a(x)$ is a polynomial of minimal
degree in $C$ which is not monic, $g(x)$ is a monic polynomial in $C$ of
minimal degree among all monic polynomials in $C,$ $x^{n}-1=k(x)g(x)$ in $%
R[x;\theta],$ $x^{n}-1=\overline{b(x)}~\overline{a(x)}$ in $\mathbb{F}%
_{p}[x],$ $\overline{a(x)}|g(x)\func{mod}u$ and $\dfrac{x^{n}-1}{g}%
up\in\left\langle ua\right\rangle .$
\end{enumerate}
\end{theorem}

\begin{proof}
The proof follows from Lemmas \ref{nonmonic}, \ref{monic1}, \ref{monic2} and %
\ref{monic3}.
\end{proof}

\section{\label{sec:Minimal-spanning-sets}Minimal spanning sets for skew
cyclic codes over $R$}

In this section we provide minimal generating sets for skew cyclic codes
over $R$. The generating sets will help in finding the cardinality of each
code. Moreover, they will be useful in describing an encoding algorithm for
these codes.

\begin{theorem}
\label{generating sets}Let $C$ be a nonzero skew cyclic code over the ring $%
R.$

\begin{enumerate}
\item If $C=\left\langle u\overline{a(x)}\right\rangle $ where $\overline{%
~a(x)}$ is a polynomial of minimal degree $r$ in $C$ and $x^{n}-1=\overline{%
b(x)}~\overline{a(x)}$ in $\ \mathbb{F}_{p}[x],$ then 
\begin{equation*}
\beta=\left\{ u\overline{a(x)},xu\overline{a(x)},\ldots,x^{n-r-1}u\overline{%
a(x)}\right\} ,
\end{equation*}
forms a minimal generating set for $C$ and $\left\vert C\right\vert
=p^{n-r}. $

\item If $C=\left\langle g(x)\right\rangle $ where $g(x)$ is a polynomial of
minimal degree $r$ in $C$ and $x^{n}-1=k(x)\ast g(x)$ in $R_{n},$ then 
\begin{equation*}
\beta=\left\{ g(x),xg(x),\ldots,x^{n-r-1}g(x)\right\} ,
\end{equation*}
forms a minimal generating set for $C$ and $\left\vert C\right\vert
=\left(p^{2}\right)^{n-r}.$

\item If $C=\left\langle g(x)+up(x),a(x)\right\rangle =\left\langle
g(x)+up(x),u\overline{a(x)}\right\rangle ,$ where $a(x)$ is a polynomial of
minimal degree $t$ in $C$ which is not monic, $g(x)$ is a monic polynomial
in $C$ of minimal degree $r$ among all monic polynomials in $C,$ $%
x^{n}-1=k(x)\ast g(x)$ in $R[x;\theta],$ $x^{n}-1=\overline{b(x)}~\overline{%
a(x)}$ in $\mathbb{F}_{p}[x],$ $\overline{a(x)}|g(x)\func{mod}u$ and $\dfrac{%
x^{n}-1}{g}up\in\left\langle ua\right\rangle .$ Then 
\begin{equation*}
\beta=\left\{ 
\begin{array}{c}
g(x)+up(x),~x\ast \left(g(x)+up(x)\right),\ldots,x^{n-r-1}\ast
\left(g(x)+up(x)\right), \\ 
u\overline{a(x)},~xu\overline{a(x)},\ldots,~x^{r-t-1}u\overline{a(x)}%
\end{array}%
\right\} ,
\end{equation*}
forms a minimal generating set for $C$ and $\left\vert C\right\vert
=\left(p^{2}\right)^{n-r}p^{r-t}.$
\end{enumerate}
\end{theorem}

\begin{proof}
We will prove Cases 1 and 3. Case 2 has similar proof.

\begin{enumerate}
\item Suppose $c(x)\in\left\langle u\overline{a(x)}\right\rangle $ where $%
\overline{~a(x)}$ is a polynomial of minimal degree $r$ in $C$ and $x^{n}-1=%
\overline{b(x)}~\overline{a(x)}$ in $\ \mathbb{F}_{p}[x].$ Then $c(x)=s(x)u%
\overline{a(x)}.$ Note that if $s(x)=s_{1}(x)+us_{2}(x),$ then $s(x)u%
\overline{a(x)}=\left(s_{1}(x)+us_{2}(x)\right)u\overline{a(x)}=s_{1}(x)u%
\overline{a(x)}+us_{2}(x)u\overline{a(x)}=s_{1}(x)u\overline{a(x)}%
+u^{2}s_{3}(x)\overline{a(x)}=s_{1}(x)u\overline{a(x)}.$ Hence, we may
assume that $s(x)=s_{1}(x)\in\mathbb{F}_{p}[x]$ and $c(x)=s_{1}(x)u\overline{%
a(x)}=us_{3}(x)\overline{a(x)}$ (by Lemma \ref{partaker}) where $\deg
s_{1}(x)=\deg s_{3}(x).$ If $\deg s_{1}(x)\leq n-r-1,$ then $c(x)=s(x)u%
\overline{a(x)}\in$span$(\beta).$ Otherwise, by the division algorithm there
are unique polynomials $q(x),~r(x)$ such that 
\begin{equation*}
s_{3}(x)=q(x)\dfrac{x^{n}-1}{\overline{a(x)}}+r(x), 
\end{equation*}
where $r(x)=0$ or $\deg r(x)<\deg\dfrac{x^{n}-1}{\overline{a(x)}}=n-r.$
Hence, 
\begin{eqnarray*}
c(x) & = & s(x)u\overline{a(x)}=s_{1}(x)u\overline{a(x)}=us_{3}(x)\overline{%
a(x)} \\
& = & u\left(q_{1}(x)\dfrac{x^{n}-1}{\overline{a(x)}}+r(x)\right)\overline{%
a(x)} \\
& = & uq_{1}(x)\dfrac{x^{n}-1}{\overline{a(x)}}\overline{a(x)}+ur(x)%
\overline{a(x)} \\
& = & ur(x)\overline{a(x)} \\
& = & r^{\prime}(x)\overline{ua(x)}
\end{eqnarray*}

where $\deg r^{\prime}(x)=\deg r(x)<\deg\dfrac{x^{n}-1}{\overline{a(x)}}=n-r.
$ Hence, $\beta$ spans the code $C.$ From the construction of the elements
in the set $\beta,$ it is clear that none of the elements is a linear
combination of the others. Therefore, $\beta$ forms a minimal generating set
for $C.$ Since $s_{1}(x)\in\mathbb{F}_{p}[x]$, we get that $\left\vert
C\right\vert =p^{n-r}.$

\item The proof is similar to Case 1.

\item Suppose that $c(x)\in C=\left\langle g(x)+up(x),a(x)\right\rangle
=\left\langle g(x)+up(x),u\overline{a(x)}\right\rangle .$ Then $%
c(x)=s_{1}(x)\ast \left(g(x)+up(x)\right)+s_{2}(x)\ast u\overline{a(x)}.$ If $\deg
s_{1}(x)\leq n-r-1,$ then $s_{1}(x)\ast \left(g(x)+up(x)\right)\in$ span($\beta$).
Otherwise, by Theorem \ref{Division} 
\begin{equation*}
s_{1}(x)=q(x)\left(\frac{x^{n}-1}{g(x)}\right)+r(x), 
\end{equation*}
where $r(x)=0$ or $\deg r(x)\leq n-r-1.$ Hence, 
\begin{eqnarray*}
s_{1}(x)\ast \left(g(x)+up(x)\right) & = & \left(q(x)\left(\frac{x^{n}-1}{g(x)}%
\right)+r(x)\right)\ast \left(g(x)+up(x)\right) \\
& = & q(x)\left(\frac{x^{n}-1}{g(x)}\right)\ast \left(g(x)+up(x)\right)+r(x)%
\ast \left(g(x)+up(x)\right) \\
& = & q(x)\left(\frac{x^{n}-1}{g(x)}\right)\ast up(x)+r(x)\ast \left(g(x)+up(x)\right)
\\
& = & uq_{1}(x)p(x)+r(x)\ast \left(g(x)+up(x)\right).
\end{eqnarray*}
Hence, 
\begin{eqnarray*}
c(x) & = & s_{1}(x)\ast \left(g(x)+up(x)\right)+s_{2}(x)\ast u\overline{a(x)} \\
& = & uq_{1}(x)p(x)+r(x)\ast\left(g(x)+up(x)\right)+s_{2}(x)\ast u\overline{a(x)} \\
& = & uq_{1}(x)p(x)+r(x)\ast \left(g(x)+up(x)\right)+us_{2}^{\prime}(x)\overline{%
a(x)} \\
& = & u\left(q_{1}(x)p(x)+s_{2}^{\prime}(x)\overline{a(x)}%
\right)+r(x)\ast \left(g(x)+up(x)\right)
\end{eqnarray*}
Since $r(x)=0$ or $\deg r(x)\leq n-r-1,$  $r(x)\ast\left(g(x)+up(x)\right)\in
$ span($\beta$). Hence, we only need to show that $uk(x)\in$ span($\beta$) for
any $uk(x)\in C.$ Suppose that $uk(x)\in C$ where $\deg k(x)\geq\deg g(x).$
Then 
\begin{equation*}
k(x)=q_{2}(x)g(x)+r_{2}(x), 
\end{equation*}
where $r_{2}(x)=0$ or $\deg r_{2}(x)<\deg g(x)$ and $\deg q_{2}(x)=\deg
k(x)-r\leq n-r-1.$ Hence $uk(x)=uq_{2}(x)g(x)+ur_{2}(x)=q_{2}^{%
\prime}(x)ug(x)+ur_{2}(x)=q_{2}^{\prime}(x)u\left(g(x)+up(x)%
\right)+ur_{2}(x).$ Since $q_{2}^{\prime}(x)u\left(g(x)+up\right)\in$span($%
\beta$),  it suffices to show that $ur_{2}(x)\in$span($\beta$) where $%
r_{2}(x)=0$ or $\deg r_{2}(x)<\deg g(x)$ and $\deg r_{2}(x)\geq\deg\overline{%
a(x)}.$ By the proof of Lemma \ref{monic3}, we know that any element of the
form $ur_{2}(x)$ belongs to $\left\langle ua(x)\right\rangle .$ Hence, $%
ur_{2}(x)=s_{4}(x)ua(x)$ and $\deg ur_{2}(x)<\deg g(x)$ and $\deg
ur_{2}(x)\geq\deg\overline{a(x)}$ $.$ Thus 
\begin{equation*}
ur_{2}(x)=\alpha_{0}u\overline{a(x)}+\alpha_{1}xu\overline{a(x)}%
+\ldots+\alpha_{r-t-1}x^{r-t-1}u\overline{a(x)}. 
\end{equation*}
Therefore, $\beta$ spans $C.$ Since none of the elements in $\beta$ is a
linear combinations of the other elements, we conclude that $\beta$ is a
minimal generating set for the code $C$ and $\left\vert C\right\vert
=\left(p^{2}\right)^{n-r}p^{r-t}.$
\end{enumerate}
\end{proof}

\section{\label{sub:The-Encoding-and}The Encoding and Decoding of the codes}

Based on Theorem \ref{generating sets}, we can develop an encoding algorithm
for these codes as follows:

\begin{theorem}
\label{class1} Let $C$ be a nonzero skew cyclic code over the ring $R.$

\begin{enumerate}
\item If $C=\left\langle u\overline{a(x)}\right\rangle $ where $\overline{%
~a(x)}$ is a polynomial of minimal degree $r$ in $C$ and $x^{n}-1=\overline{%
b(x)}~\overline{a(x)}$ in $\ \mathbb{F}_{p}[x],$ then any codeword $c(x)$ in 
$C$ ~is encoded as 
\begin{equation*}
c(x)=i(x)u\overline{a(x)},
\end{equation*}
where $i(x)\in\mathbb{F}_{p}[x]$ is a polynomial of degree $\leq$ $n-r-1.$

\item If $C=\left\langle g(x)\right\rangle $ where $g(x)$ is a polynomial of
minimal degree $r$ in $C$ and $x^{n}-1=k(x)g(x)$ in $R_{n},$ then any
codeword $c(x)$ in $C$ ~is encoded as 
\begin{equation*}
c(x)=(i(x)+uq(x))g(x),
\end{equation*}
where $i(x)+uq(x)\in R[x;\theta]$ is a polynomial of degree $\leq n-r-1$.

\item Let $C=\left\langle g(x)+up(x),a(x)\right\rangle =\left\langle
g(x)+up(x),u\overline{a(x)}\right\rangle ,$ where $a(x)$ is a polynomial of
minimal degree $\tau$ in $C$ which is not monic, $g(x)$ is a monic
polynomial in $C$ of minimal degree $r$ among all monic polynomials in $C,$ $%
x^{n}-1=k(x)g(x)$ in $R[x;\theta],$ $x^{n}-1=\overline{b(x)}~\overline{a(x)}$
in $\mathbb{F}_{p}[x],$ $\overline{a(x)}|g(x)\func{mod}u$ and $\dfrac{x^{n}-1%
}{g}up\in\left\langle ua\right\rangle .$ Then, any codeword $c(x)$ in $C$ is
encoded as 
\begin{equation*}
c(x)=\left(i\left(x\right)+uq(x)\right)\left(g(x)+up(x)\right)+j(x)u%
\overline{a(x)},
\end{equation*}
where $i(x)+uq(x)\in R[x;\theta]$ is a polynomial of degree $\leq n-r-1$ and 
$j(x)\in\mathbb{F}_{p}[x]$ is a polynomial of degree$\leq r-\tau-1.$
\end{enumerate}
\end{theorem}

\begin{proof}
The proof follows from Theorem \ref{generating sets}.
\end{proof}


Now, we describe a decoding algorithm when one considers the case (iii) in
Theorem \ref{class1}.


Suppose that a string of $r-\tau$ symbols $J=(j_{0},j_{1},\cdots,j_{r-%
\tau-1})\in\mathbb{F}_{p}^{r-\tau}$ and a string of symbols $I=\big(%
(i_{0}+uq_{0}),\cdots,(i_{n-r-1}+uq_{n-r-1})\big)\in\big(\mathbb{F}_{p}+u%
\mathbb{F}_{p}\big)^{n-r}$ are the inputs of the encoder. So, according to
Theorem \ref{class1}, the encoding process will be as follows 
\begin{align}
encod(J,I)=\Big(uj(x)\overline{a(x)}+\big(i(x)+uq(x)\big)\big(g(x)+up(x)\big)%
\Big)(mod\quad x^{n}-1).  \label{encod}
\end{align}
The encoded string is of length $n$ symbols over $\mathbb{F}_{p}+u\mathbb{F}%
_{p}$ which is transmitted through the channel.


The decoder receives a string of length $n$ over $\mathbb{F}_{p}+u\mathbb{F}%
_{p}$. Let us represent that string by $V=\Big((\alpha _{0},c_{0}),\cdots
,(\alpha _{n-1},c_{n-1})\Big)$. Define $v(x)=\sum_{i=0}^{n-1}(\alpha
_{i}+uc_{i})x^{i}$. Since $g+up$ is a monic polynomial, we may observe from
Lemma \ref{Division} that there exist $l_{1}+ul_{2},r_{1}+ur_{2}\in R_{n}$
such that $v=(l_{1}+ul_{2})(g+up)+r_{1}+ur_{2}$ and $\deg
(r_{1}+ur_{2})<\deg (g+up)$ . Also, there exist $s,t\in \frac{\mathbb{F}%
_{p}[x]}{<x^{n}-1>}$ such that $r_{2}=t\overline{a}+s$. Thus, $%
v=(l_{1}+ul_{2})(g+up)+r_{1}+ut\overline{a}+us$. This results in the fact
that $r_{1},s$ are the error terms. Therefore, $\widehat{v}%
=l_{1}g+ul_{2}g+ul_{1}^{\prime }p+u\overline{a}t$ is a codeword. Our goal is
to find $j(x),(i+uq)(x)$ which are introduced in the encoding part. By a
comparison between $\widehat{v}$ and Equation (\ref{encod}), we observe that 
$l_{1}=i$ and $gl_{2}+l_{1}^{\prime }p+\overline{a}t=qg+i^{\prime }p+%
\overline{a}j$ (Note that these equations are in the ring $\frac{\mathbb{F}%
_{p}[x]}{<x^{n}-1>}$). Since $l_{1}=i$, we have $i^{\prime }p=l_{1}^{\prime
}p$. Therefore, we need to solve the equation $qg+\overline{a}j=l_{2}g+%
\overline{a}t$ to find $q,j$. Since $\overline{a}|g$, we may assume that $g=%
\overline{a}\overline{d}$. Therefore, $q\overline{d}+j=l_{2}\overline{d}+t$.
Thus, $j,q$ are in the form $t-\overline{d}z,l_{2}+z$ for some $z\in \frac{%
\mathbb{F}_{p}[x]}{x^{n}-1}$, respectively. Hence, there exists $A,B\in 
\mathbb{F}_{p}[x]$ such that 
\begin{align}
j(x)=& t(x)-\overline{d(x)}z(x)+A(x)(x^{n}-1) \\
q(x)=& l_{2}(x)+z(x)+B(x)(x^{n}-1).
\end{align}%
If $\overline{a}\overline{d}h=gh=x^{n}-1$, then 
\begin{align}
j(x)=& t(x)-\overline{d(x)}(z(x)+A(x)\overline{a(x)}h(x)) \\
q(x)=& l_{2}(x)+z(x)+B(x)(x^{n}-1).
\end{align}%
We know that $\deg (q)\leq k,\deg (j)\leq T$. Since $\deg (d)=\deg (g)-\deg
(a)>T$, $j$ can be uniquely found as $j=t-\overline{d}(z+A\overline{a}h)$,
which leads to finding$z+A\overline{a}h$. Since every $A,B$ leads to the
same answer for $j,q$ in $R_{n}$, we may assume $A=B=1$. Thus, $q=l_{2}+%
\frac{t-j}{\overline{d}}-\overline{a}h+(x^{n}-1)$. The outputs of the
decoding process are $j(x)$ and $i(x)+uq(x)$.

\section{\label{sec:Examples}Examples}

In this section, we provide examples of encoding and decoding of skew cyclic
codes.

\begin{example}
In this example, we see the steps to encode and decode with the proposed
algorithm. Let $n=6,r=3,\tau=2$ $C=<x^{3}+2x^{2}+(2+u)x+1+u,u(x+1)>$. The
number of symbols of output is $12$ and the number of symbols of input $8$.
Suppose that the encoder wants to transmit two strings $I=(1+u,3u,4+2u)\in(%
\mathbb{F}+p+u\mathbb{F}_{p})^{3}$ and $J=(2,2)\in(\mathbb{F}_{p})^{2}$. So 
\begin{align*}
& Encode((1+u,3u,4+2u),(3,2))= \\
& \quad((1+u)x^{2}+3ux+(4+2u))(x^{3}+2x^{2}+(2+u)x+1+u)+u(2x+2)(x+1) \\
& \quad=(1+u)x^{5}+2x^{4}+(1+u)x^{3}+(4+4u)x^{2}+3x+4+3u.
\end{align*}
So the encoder sends $(1+u,2,1+u,4+4u,3,4+3u)$ through the channel.

Now assume that the decoder receives $(1+2u,2,1+u,4+4u,3,4+3u)$. If $%
v(x)+uw(x)=(1+2u)x^{5}+2x^{4}+(1+u)x^{3}+(4+4u)x^{2}+3x+4+3u$, then the
syndromes $s_{1}(x),s_{2}(x)$ are as follows. 
\begin{align*}
& s_{1}(x)=v(x)(x^{3}-2x^{2}+2x-1)=0 \\
s_{2}(x)=& w(x)(x^{5}-x^{4}+x^{3}-x^{2}+x-1)=x^{5}-x^{4}+x^{3}-x^{2}+x-1.
\end{align*}%
One can see that 
\begin{equation*}
x^{5}(x^{5}-x^{4}+x^{3}-x^{2}+x-1)=x^{5}-x^{4}+x^{3}-x^{2}+x-1.
\end{equation*}%
So $s_{2}$ is equal to the syndrome of $x^{5}$. This means that the decoder
detect the error term $ux^{5}$. It remains to find $I,J$. To do this, we
have to divide $(1+u)x^{5}+2x^{4}+(1+u)x^{3}+(4+2u)x^{2}+(3+u)x+4+3u$ by $%
x^{3}+2x^{2}+(2+u)x+1+u$. One can see that 
\begin{align*}
& (1+u)x^{5}+2x^{4}+(1+u)x^{3}+(4+2u)x^{2}+(3+u)x+4+3u \\
& =((1+u)x^{2}+3ux+2u+4)(x^{3}+2x^{2}+(2+u)x+1+u) \\
& +u(2x^{2}+4x+2).
\end{align*}%
Since $2x^{2}+4x+2=(2x+2)(x+1)$, we can extract $I,J$. 
\end{example}

Finally, we propose an example of the decoding process in section \ref%
{sub:The-Encoding-and}.

\begin{example}
We want to study the principal skew cyclic codes with length $4$ over $%
F_{3}+uF_{3}$ with $\theta(u)=-u$. For this, one can see that the
composition of $x^{4}-1$ in $F_{3}$ is as follows. 
\begin{align}
x^{4}-1=(x+2)(x+1)(x^{2}+1)
\end{align}
There are two types of principal codes. If the generator is not monic, then $%
C=\left\langle u\overline{a(x)}\right\rangle $. Therefore, all of the codes
in this form are as follows. 
\begin{align*}
C_{1}= & <u(x+2)>\quad or\quad C_{2}=<u(x+1)>\quad or\quad C_{3}=<u(x^{2}+1)>
\\
C_{4}= & <u(x^{2}+2)>\quad or\quad C_{5}=<u(x^{3}+x^{2}+x+1)> \\
C_{6}= & <u(x^{3}+2x^{2}+x+2)>\quad or\quad C_{7}=<u>\quad or\quad C_{8}=<0>.
\end{align*}
The generator matrix $G$ and the parity check matrix $H$ of the code $C_{4}$
are as follows. 
\begin{align*}
G=%
\begin{bmatrix}
u & 0 & 2u & 0 \\ 
0 & 2u & 0 & u%
\end{bmatrix}%
\quad\quad H=%
\begin{bmatrix}
0 & 1 & 0 & 1 \\ 
1 & 0 & 1 & 0%
\end{bmatrix}%
\end{align*}
On the other hand, if the generator of the code is monic, then $%
C=\left\langle g(x)+up(x)\right\rangle $ $x^{n}-1=k\ast g$ for some $k\in
R_{n}$. An example of such code is $C=\langle (1+u)x^{2}+1+2u\rangle$. The
generator matrix and parity check matrix of this code are as follows. 
\begin{align*}
G=%
\begin{bmatrix}
1+u & 0 & 1+2u & 0 \\ 
0 & 1+2u & 0 & 1+u%
\end{bmatrix}%
\quad\quad H=%
\begin{bmatrix}
0 & 2+2u & 0 & 2+2u \\ 
1+u & 0 & 1+u & 0%
\end{bmatrix}%
\end{align*}
\end{example}

\section{\label{sec:Examples} Gray Images and Codes with Good Parameters}

One of our goals in this study was to obtain codes with good parameters over 
$\mathbb{F}_{p}$ from skew cyclic codes over $R$. To this end, we need a map
from $R$ to $\mathbb{F}_{p}^{\ell }$ for some positive integer $\ell $. We
use the map given in \cite{GrayMap}. For any integer $\ell ,~1\leq \ell \leq
p,$ define the Gray mapping as%
\begin{eqnarray*}
\varphi _{\ell } &:&R\rightarrow \mathbb{F}_{p}^{\ell } \\
\varphi _{\ell }\left( a+ub\right)  &=&\left( b,b+a,b+2a,\ldots b+\left(
\ell -1\right) a\right) .
\end{eqnarray*}

This map is naturally extended to a map $\varphi_{\ell}$ from $R^{n}$ to $%
\mathbb{F}_{p}^{\ell n}$.

For $c=a+ub\in R$, define the Gray weight of $c$ to be%
\begin{equation*}
w\left( c\right) =\left\{ 
\begin{array}{c}
0\text{ \ \ \ if }c=0 \\ 
\ell-1\text{ \ \ if }c=x-u\lambda x,~x\in \mathbb{F}_{p}^{\ast },~0\leq
\lambda \leq \ell-1 \\ 
\ell \text{ \ \ \ \ \ \ \ \ \ \ \ otherwise}%
\end{array}%
\right. 
\end{equation*}

It is shown in \cite{GrayMap} that $\varphi_{\ell}$ is a linear, distance
preserving map from $R^{n}$ to $\mathbb{F}_{p}^{\ell n}$. It is one-to-one
if $\ell \geq 2$. Therefore, if $C$ is a linear code over $R$ with
parameters $(n,M,d)$ where $d$ is the minimum Gray weight of $C$, then for $%
\ell \geq 2$ $\varphi_{\ell}(C)$ is a linear code over $\mathbb{F}_p$ with
parameters $(n\ell,M,d)$. If $C$ is a free code of dimension $k$ over $R$,
then $\varphi_{\ell}(C)$ is a linear code of dimension $2k$ over $\mathbb{F}%
_p$.

We searched over skew cyclic codes with generators of the form (2) in
Theorem \ref{classifications}. Hence they are free codes over $R$ with
dimension $k=n-\deg(g(x))$ where $g(x)$ is a divisor of $x^n-1$ in $%
R[x;\theta]$. Then we applied the Gray map described above to obtain linear
codes over $\mathbb{F}_p$. As a result of a computer search which is carried
out using Magma software, we obtained a number of codes with optimal or near
optimal parameters. We list below a sample of these codes.

\begin{example}
Let $p=3, n=8$. The polynomial $g=x^3 + ux^2 + x + 1$ divides $x^8-1$ over $%
R=\mathbb{F}_3+u\mathbb{F}_3$, hence it generates a free cyclic code of
dimension 5 over $R$. Its image $\varphi_2(C)$ is a ternary linear code with
parameters $[16,10,4]$ which, according to the database \cite{site}, is an
optimal code over $\mathbb{F}_3$.
\end{example}

\begin{example}
Let $p=3, n=6$. The polynomial $g=x^4 + 2x^3 + 2ux^2 + x + u + 2$ divides $%
x^6-1$ over $R=\mathbb{F}_3+u\mathbb{F}_3$, hence it generates a free cyclic
code of dimension 2 over $R$. Its image $\varphi_2(C)$ is a ternary linear
code with parameters $[12,4,6]$ and $\varphi_3(C)$ is a ternary linear code
with parameters $[18,4,11]$. Both of these codes are optimal according to 
\cite{site}.
\end{example}

\begin{example}
Let $p=3, n=11$. The polynomial $g=x^6 + x^4 + 2x^3 + 2x^2 + 2x + 1$ divides 
$x^{11}-1$ over $R=\mathbb{F}_3+u\mathbb{F}_3$, hence it generates a free
cyclic code of dimension 5 over $R$. Its image $\varphi_2(C)$ is a ternary
linear code with parameters $[22,10,6]$ which turns out to be a quasi-cyclic
code. According to the database of best known quasi-twisted codes (which
includes quasi-cyclic codes as a special case) \cite{qtDatabase}, this is a
new code in the class of quasi-twisted codes.
\end{example}

\begin{example}
Let $p=3, n=12$. The polynomial $g=x^5 + (u + 1)x^4 + ux^3 + 2ux^2 + (2u +
2)x + 2u + 2$ divides $x^{12}-1$ over $R=\mathbb{F}_3+u\mathbb{F}_3$, hence
it generates a free cyclic code of dimension 7 over $R$. Its image $%
\varphi_2(C)$ is a ternary linear code with parameters $[24,14,6]$ which,
according to \cite{site}, has the parameters of a best known ternary linear
code.
\end{example}

\begin{example}
Let $p=5, n=6$. The polynomial $g=x^4 + (3u + 4)x^3 + 4ux^2 + (2u + 1)x + u
+ 4$ divides $x^6-1$ over $R=\mathbb{F}_5+u\mathbb{F}_5$, hence it generates
a free cyclic code of dimension 2 over $R$. Its image $\varphi_3(C)$ is a
linear code with parameters $[18,4,12]$ over $\mathbb{F}_5$. According to
the database \cite{site}, this is an optimal linear code.
\end{example}

\begin{example}
Let $p=5, n=18$. The polynomial $g=x^3 + 2ux + 3u + 1$ divides $x^{18}-1$
over $R=\mathbb{F}_5+u\mathbb{F}_5$, hence it generates a free cyclic code
of dimension 15 over $R$. Its image $\varphi_2(C)$ is a linear code with
parameters $[36,30,4]$ over $\mathbb{F}_5$. According to the database \cite%
{site}, these are the parameters of the best known linear code for its
length and dimension over $\mathbb{F}_5$.
\end{example}

\section{\label{sec:Conclusion}Conclusion}

In this paper we studied skew cyclic codes over the ring $R=\mathbb{F}_{p}%
\mathbb{+}u\mathbb{F}_{p}$ where $p$ is an odd prime and $u^{2}=0.$ We have
classified skew cyclic codes of all lengths as left $R[x;\theta ]$%
-submodules of $R_{n}=R[x;\theta ]/\langle x^{n}-1\rangle .$ Our
classification is general and works for any value of $n$. Then we
constructed the set of generator polynomials for these codes. We also
provided encoding and decoding algorithms for these codes. Additionally, end
we presented examples of skew cyclic codes whose Gray images are linear
codes over $\mathbb{F}_{p}$ with optimal or near optimal parameters. 
We also provided a new code in the class of quasi-twisted codes.


\end{document}